\documentclass[aps,prx,twocolumn,superscriptaddress,export]{revtex4-1} 
\usepackage[colorlinks=true,linkcolor=blue,citecolor=blue,urlcolor=blue]{hyperref}

\usepackage[utf8]{inputenc}
\usepackage[german, english]{babel}
\usepackage[dvipsnames]{xcolor}
\usepackage{bbm}
\usepackage{times}
\usepackage{tikz-cd}
\hypersetup{colorlinks=true}
\usetikzlibrary{positioning}
\usetikzlibrary{decorations,arrows}
\usetikzlibrary{decorations.pathreplacing}
\usetikzlibrary{chains, scopes, positioning, backgrounds, shapes, fit, shadows, calc, arrows.meta}
\tikzset{>=stealth}
\usetikzlibrary{through,hobby}
\usetikzlibrary{decorations.markings}
\usetikzlibrary{fadings,shadings}
\usetikzlibrary{plotmarks}
\usepackage{dcolumn}  
\usepackage{bm}        
\usepackage{bbold}
\usepackage{amsthm}
\usepackage{mathtools}
\usepackage{enumerate}
\usepackage{tikz}
\usetikzlibrary{shapes.geometric, arrows, shapes.multipart}
\usepackage{booktabs, tabularx}

\usepackage{blkarray}
\usepackage{xfrac}

\usepackage[export]{adjustbox}
\usepackage{titlesec}
\usepackage{amsfonts}
\usepackage{amssymb}
\usepackage{braket}

\usepackage{subfigure}

\newcommand{\tr}{\mathrm{Tr}}
\newcommand{\Z}{\mathbb{Z}}
\newcommand{\E}{\mathcal{E}}

\newcommand{\diff}[1]{{{#1}}}

\newcommand{\id}{\mathbb{1}}
\newcommand{\spt}{\mathrm{SPT}}
\newcommand{\triv}{\mathrm{triv.}}
\newcommand{\CZ}{{\mathrm{CZ}}}

\newtheorem{Theorem}{Theorem}

\usepackage[normalem]{ulem}

\newtheorem{Cor}{Corollary}

\usepackage{amsthm}
\theoremstyle{plain}
\newtheorem*{theorem*}{Theorem}

\newtheorem{conjecture}{Conjecture}
\newtheorem{definition}{Definition}

\begin{document}

\title{Many-body quantum catalysts for transforming between phases of matter}

\date{\today}

\author{David T. Stephen}
\affiliation{Department of Physics and Center for Theory of Quantum Matter, University of Colorado Boulder, Boulder, CO 80309, USA}
\affiliation{Department of Physics, California Institute of Technology, Pasadena, CA 91125, USA}
\affiliation{Quantinuum, 303 S Technology Ct, Broomfield, CO 80021, USA}
\author{Rahul Nandkishore}
\affiliation{Department of Physics and Center for Theory of Quantum Matter, University of Colorado Boulder, Boulder, CO 80309, USA}
\author{Jian-Hao Zhang}
\email{Sergio.Zhang@colorado.edu}
\affiliation{Department of Physics and Center for Theory of Quantum Matter, University of Colorado Boulder, Boulder, CO 80309, USA}

\begin{abstract}
A catalyst is a substance that enables otherwise impossible transformations between states of a system, without being consumed in the process. In this work, we apply the notion of catalysts to many-body quantum physics. In particular, we construct catalysts that enable transformations between different symmetry-protected topological (SPT) phases of matter using symmetric finite-depth quantum circuits. We discover a wide variety of catalysts, including GHZ-like states that spontaneously break the symmetry, gapless states with critical correlations, topological orders with symmetry fractionalization, and spin-glass states. These catalysts are all united under a single framework that has close connections to the theory of quantum anomalies, and we use this connection to put strong constraints on possible pure- and mixed-state catalysts.
We also show how the catalyst approach leads to new insights into the structure of certain phases of matter, and to new methods to efficiently prepare SPT phases with long-range interactions or projective measurements. 
\end{abstract}

\maketitle

\section{Introduction}

In chemistry, a catalyst is a substance that enables or speeds up a chemical reaction that would otherwise be impossible or slow. Furthermore, the catalyst is not consumed in the process, so it can be reused many times. The idea of catalysts has also been successfully applied to few-body quantum mechanics. For example, an entanglement catalyst is a quantum state that enables otherwise impossible transformations between two states using local operations and classical communication \cite{Jonathan1999}, and this perspective has led to new insights into entanglement theory \cite{Kondra2021}. Catalysts have also been applied to more general quantum resources theories \cite{Anshu2018, Marvian2022} and the generation of magic states for fault-tolerant quantum computation \cite{Gidney2019efficientmagicstate,fang2024surpassing}. In general, understanding the role of catalysts leads to more efficient methods of achieving a desired transformation, and also in many cases provides new insights into the fundamental properties of said transformation.

In this paper, we apply catalysts to the setting of many-body quantum physics. In this setting, spatial locality plays a crucial role, and a natural kind of transformation is that which can be generated in finite time by a local Hamiltonian. Here, ``finite'' means independent of system size, allowing us to talk about transformations between states in the thermodynamic limit. It is known that such transformations are equivalent to finite-depth quantum circuits (FDQCs) \cite{Chen2010}. Whether or not two many-body states can be related via FDQC is a question of both fundamental and practical importance. On one hand, the classification of topological phases of matter coincides exactly with the equivalence classes of many-body ground states under FDQCs, with a trivial state being one that can be prepared from an unentangled states using an FDQC \cite{Chen2010}. On the other hand, FDQCs are the operation that is most naturally implemented in a quantum computer, so any transformation implemented by FDQC is said to be ``easy''. 

\begin{figure}
    \centering
    \includegraphics[width=0.99\linewidth]{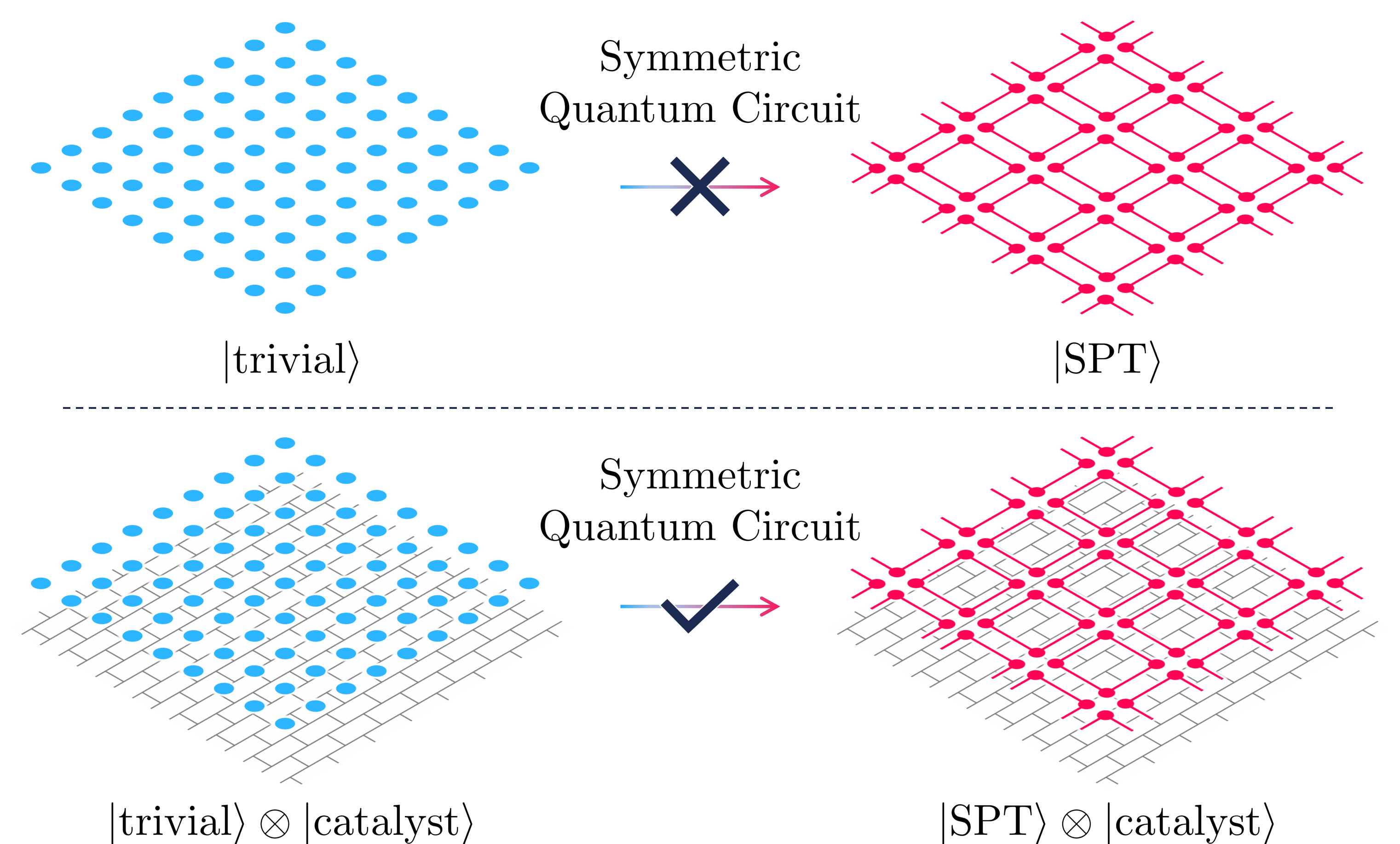}
    \caption{Certain many-body quantum states, such as those belonging to different SPT phases of matter, cannot be related using symmetric finite-depth quantum circuits. This changes in the presence of an appropriate long-range entangled many-body catalyst.}
    \label{fig:main}
\end{figure}

One can additionally impose symmetries on the Hamiltonian or circuit, and a state that can be mapped to a product state using an FDQC, but cannot be using a symmetric FDQC, is said to belong to a non-trivial symmetry-protected topological (SPT) phase of matter \cite{Chen2013}. This is the class of states that we focus on in this work. In particular, we construct catalysts that enable transformations between distinct SPT phases of matter using symmetric FDQCs, something that is impossible without a catalyst, see Fig.~\ref{fig:main}. These catalysts \diff{are themselves many-body states matching the dimensionality of the system, and they} have widely varying physical properties including symmetry-breaking states with cat-like entanglement, gapless ground states of conformal field theories, topologically ordered states with fractionalized excitations, and disordered spin-glass states. All of these types of catalyst are unified under a single framework. Namely, if a transformation between SPT phases with symmetry group $G$ can be implemented by a symmetric quantum cellular automaton $\mathcal{U}$, then any state that is invariant under both $G$ and $\mathcal{U}$ can catalyze the transformation by symmetric FDQC. Remarkably, the combination of $G$ and $\mathcal{U}$ symmetry has deep connections to the notion of quantum anomalies which, in particular, implies that pure-state catalysts must have long-range entanglement (although we show this is not true for mixed-state catalysts).

Using our framework, we construct explicit examples of catalysts and their associated symmetric FDQCs to transform between a large variety of SPT phases of matter \diff{with on-site unitary symmetry groups}. \diff{In general, pure state catalysts are at least as difficult to prepare as the SPT phases themselves. However, since the catalyst can be reused, this hard work only needs to be done once, since a single catalyst can be used to prepare many SPT states. On the other hand, we show that mixed-state catalysts can be prepared in finite time using symmetric local channels.} Furthermore in some cases, the form of the catalysts helps us to learn more about physical properties the SPT phase at hand. We also show how the catalyzed transformation can be used to construct new methods of efficient state preparation using, for example, long-range power-law decaying interactions or projective measurements. 
Overall, our results pave the way for studying quantum catalysts in the context of many-body quantum systems.

The rest of the paper is organized as follows. In Sec.~\ref{sec:catalysts}, we define our notion of many-body catalysts, and present Theorem.~\ref{thm:catalyst} which gives a generic recipe to construct catalysts for SPT phases. We then discuss a simple example of these ideas which is intimately related to the famous Lieb-Schultz-Mattis (LSM) anomaly \cite{Lieb1961}. In Sec.~\ref{sec:global_spt} we construct catalysts that transform between all SPT phases with global symmetries that lie within the so-called group cohomology classification \cite{Chen2013}. In Sec.~\ref{sec:other_symms}, we consider higher-form and subsystem symmetries. In Sec.~\ref{sec:mixed}, we then generalize our notion of catalysts to allow mixed states. We give examples of mixed-state catalysts that relate to the recent notion of strong-to-weak symmetry breaking \cite{SWSSB2024}, and we then characterize all mixed-state catalysts in one dimension. In Sec.~\ref{sec:state_prep}, we show how the catalyst framework can lead to novel methods to efficiently prepare SPT phases. Finally, in Sec.~\ref{sec:discussion}, we discuss future directions of study for many-body catalysts.

\section{Many-body catalysts} \label{sec:catalysts}

In this section, we define our notion of many-body quantum catalysts, give a systematic procedure to construct a large class of catalysts, and present a simple case study that illustrates the diversity of this class.

\subsection{Definitions}

First, we establish some terminology. A \textit{quantum cellular automaton (QCA)} is a unitary $\mathcal{U}$ that preserves locality \cite{schumacher2004,Farrelly2020reviewofquantum}. That is, for any operator $O$ supported in a ball of radius $r$, the conjugated operator $\mathcal{U}O\mathcal{U}^\dagger$ is supported in a ball of radius $r+d$ for some constant $d$. A typical example of a QCA is a \textit{finite-depth quantum circuit (FDQC)}, which is a unitary operator of the form,
\begin{equation}
    \mathcal{U} = \prod_{\ell=1}^D\left( \prod_i u_{\ell,i}\right),
\end{equation}
where each $u_{\ell,i}$, referred to as a gate, acts only on qubits within a ball of radius $r$ and the gates within a given layer $\ell$ have non-overlapping support. Both $D$ and $r$ should furthermore be independent of system size. An example of a QCA beyond FDQCs is the lattice translation operator $\mathcal{T}$. Given some symmetry operators $U(g)$ for $g$ in some group $G$, we define a \textit{$G$-symmetric QCA} to be a QCA that commutes with $U(g)$ for all $g\in G$. A \textit{$G$-symmetric FDQC} is similarly an FDQC where each gate $u_{\ell,i}$ commutes with $U(g)$ for all $g\in G$. We remark that it is possible for a symmetric QCA to be an FDQC but not a symmetric FDQC, meaning that the operator is symmetric as a whole but cannot be written in terms of local symmetric gates.

\diff{In this paper, we consider bosonic SPT phases with an on-site unitary symmetry $U(g)=\prod_i u_i(g)$ for $g\in G$ (which is not necessarily spatially uniform). The objects of interest are symmetric quantum states which satisfy $U(g)|\psi\rangle = |\psi\rangle$ for all $g\in G$.} The classification of SPT phases is based on equivalence classes of symmetric quantum states modulo symmetric FDQCs \cite{Chen2013}. That is, two symmetric states $|\psi_1\rangle$ and $|\psi_2\rangle$ are in the same SPT phase if and only if there is a symmetric FDQC $\mathcal{U}$ such that 
\begin{equation} \label{eq:spt_equiv}
    |\psi_1\rangle = \mathcal{U}|\psi_2\rangle.
\end{equation}
If we define the trivial SPT phase to be the phase containing symmetric product states, then the above definition implies that a state belonging to a non-trivial SPT phase can never be prepared from a symmetric product state using a symmetric FDQC. Instead, one requires a symmetric quantum circuit whose depth grows linearly with system size \cite{Huang2015,Chen2024sequential} or a symmetric circuit with long-range gates \cite{Stephen2024klocal}.

Typically, one also extends the above definition to allow the addition of ancillary degrees of freedom in symmetric product states to either side of Eq.~\eqref{eq:spt_equiv}. This allows the comparison of states on different system sizes, different lattices, and with different local Hilbert spaces (\textit{e.g.} spin-1/2 and spin-1 particles). One way to understand the results of this paper is that we study how the classification of SPT phases changes when non-trivial ancillae are allowed.

We define our notion of many-body catalysts as follows,
\begin{definition}[Many-body catalysts] \label{def:catalyst}
    Given a pair of $G$-symmetric states $|\psi_0\rangle$ and $|\psi_1\rangle$, we say that a $G$-symmetric state $|a\rangle$ catalyzes the transformation $|\psi_0\rangle\mapsto|\psi_1\rangle$ if there exists a $G$-symmetric FDQC $\mathcal{V}$ such that,
    \begin{equation}
        \mathcal{V}(|\psi_0\rangle\otimes |a\rangle) = |\psi_1\rangle\otimes |a\rangle.
    \end{equation}
\end{definition}
\noindent
Observe that we only consider catalysts which are themselves symmetric such that the symmetry of the whole system is preserved. \diff{Furthermore, the dimensionality of the catalyst should match that of the system that is being catalyzed, i.e. the catalyst should be defined on a constant number of ancillary degrees of freedom per site $i$. If one allows higher dimensional catalysts, for example, then the pumping construction of Ref.~\cite{Tantivasadakarn2022pump} shows that even a symmetric product state can serve as a catalyst.}

Here we considered only pure-state catalysts, but in Sec.~\ref{sec:mixed} will also allow mixed-state catalysts which are defined using a suitable generalization of Definition~\ref{def:catalyst}. Naturally, since the definition of many-body catalysts include the notion of an FDQC, it should only be applied in a scenario where the initial, final, and catalyst states can be defined on systems of increasing size, and the depth $D$ and range $r$ of the FDQC are both independent of system size.

To understand this definition operationally, we imagine a resource theory-like scenario where local, symmetric interactions are easy to implement, but other kinds of interactions are hard. Consider the case where $|\psi_0\rangle = |\triv\rangle$, a trivial symmetric product state, and $|\psi_1\rangle$ belongs to a non-trivial SPT phase. Then it is ``hard'' to transform between the two states. However, in the presence of the catalyst $|a\rangle$, the transformation is ``easy'' since it can be implemented by a symmetric FDQC. Furthermore, the catalyst remains unchanged in the process. Therefore, one can imagine that, with ample supply of product states $|\triv\rangle$, SPT states can be continually ``pulled'' out of the catalyst. The goal of this paper is to characterize catalysts for SPT phases, construct examples, relate them to the physics of SPT order, and examine the applications to state preparation in quantum devices.

This operational definition also leads to a strong constraint on the form of possible catalysts. Namely, it cannot be possible to generate a catalyst for a non-trivial SPT phase by acting with a symmetric FDQC on a product state. If this was possible, then the process of creating the catalyst, using it to catalyze the SPT phase, and destroying the catalyst describes an symmetric FDQC that creates the SPT phase from a product state, which is a contradiction. More generally, the catalyst  $|a\rangle$ cannot be symmetrically invertible \cite{SWSSB2024}, meaning that there cannot exist a state $|a^{-1}\rangle$ such that $|a\rangle \otimes |a^{-1}\rangle$ is connected to a symmetric product state via a symmetric FDQC. In other words, it must be at least as hard to make the catalyst as it is to make the SPT phase on its own. However, the advantage of using a catalyst is that the hard work only needs to be performed once, since the catalyst is not consumed after use and can therefore be re-used.

\subsection{$(G,\mathcal{U})$-symmetric catalysts}

Throughout this paper, we will primarily focus on a particular class of catalysts that we call $(G,\mathcal{U})$-symmetric catalysts. This class will be well-suited to studying all kinds of SPT phases.
Their construction is based on two facts. First, any fixed-point SPT state $|\mathrm{SPT}\rangle$ within the group cohomology classification can be created from a symmetric product state $|\triv\rangle$ from a symmetric QCA $\mathcal{U}$ called the \textit{SPT entangler} \cite{Chen2013},
\begin{equation}
    |\mathrm{SPT}\rangle = \mathcal{U}|\triv\rangle\ .
\end{equation}
It is important to note that, while $\mathcal{U}$ is a symmetric QCA, it is not a symmetric FDQC.
Second, given any symmetric QCA $\mathcal{U}$, the operator $\mathcal{U}\otimes \mathcal{U}^{-1}$ on the doubled Hilbert space can be written as a symmetric FDQC with respect to the doubled symmetry $U'(g) = U(g)\otimes U(g)$ \cite{Farrelly2020reviewofquantum,Stephen2024klocal}. To do this, define the SWAP operator $\mathcal{S}$ that exchanges the two Hilbert spaces, $\mathcal{S}(|\psi\rangle\otimes |\phi\rangle)=|\phi\rangle\otimes |\psi\rangle$. This can be written as a product of local SWAP gates, $\mathcal{S}=\prod_i s_i$ where each $s_i$ exchanges site $i$ between the two Hilbert spaces. Note that $s_i$ commutes with $U'(g)$ since $U(g)$ is a product on onsite operators. Then, we can write,
\begin{equation} \label{eq:qca_symmfdqc}
\begin{aligned}
    \mathcal{U}\otimes \mathcal{U}^{-1} 
    &=  (\mathcal{U}\otimes I) \mathcal{S}^{-1} (\mathcal{U}^{-1}\otimes I) \mathcal{S} \\
    &= \prod_i v_i \prod_i s_i,
\end{aligned}
\end{equation}
where $v_i = (\mathcal{U}\otimes I) s_i (\mathcal{U}^{-1}\otimes I)$. Since $\mathcal{U}$ is locality preserving, $v_i$ is a local operator supported near sites $i$. Finally, since both $s_i$ and $\mathcal{U}\otimes I$ commute with $U'(g)$, $v_i$ does as well. Therefore, Eq.~\eqref{eq:qca_symmfdqc} is a decomposition of $\mathcal{U}\otimes \mathcal{U}^{-1}$ into local, $U'(g)$-symmetric gates which can be suitably divided into finitely-many disjoint layers, so $\mathcal{U}\otimes \mathcal{U}^{-1}$ is a symmetric FDQC.

With the above facts, we can describe our construction of SPT catalysts, which is summarized as follows,
\begin{Theorem} \label{thm:catalyst}
    Consider an arbitrary $G$-symmetric state $|\psi\rangle$ and a $G$-symmetric QCA $\mathcal{U}$. Then, any $(G,\mathcal{U})$-symmetric state can catalyze the transformation $|\psi\rangle \mapsto \mathcal{U}|\psi\rangle$.
\end{Theorem}
\begin{proof}
    Let $|a\rangle$ be a state such that $U(g)|a\rangle=|a\rangle$ for all $g\in G$ and also $\mathcal{U}|a\rangle = |a\rangle$, and let $\mathcal{V}=\mathcal{U}\otimes \mathcal{U}^{-1}$. From the above, we know that $V$ is a symmetric FDQC. Then we have,
    \begin{equation}
         \mathcal{V}(|\psi\rangle \otimes |a\rangle) = \mathcal{U}|\psi\rangle \otimes \mathcal{U}^{-1}|a\rangle = \mathcal{U}|\psi\rangle \otimes |a\rangle,
    \end{equation}
    where we used the simple observation that, if $|a\rangle$ is invariant under $\mathcal{U}$, it is also invariant under $\mathcal{U}^{-1}$.
\end{proof}
In particular, if we take $|\psi\rangle = |\triv\rangle$ and choose $\mathcal{U}$ to be the entangler of a state $|SPT\rangle$ belonging to a non-trivial SPT phase, then this result implies that any $(G,\mathcal{U})$-symmetric state can catalyze the transformation $|\triv\rangle\mapsto |SPT\rangle$.

Theorem~\ref{thm:catalyst} is an important result that will help us to construct a large class of SPT-catalysts with various physical properties. Furthermore, Theorem~\ref{thm:catalyst} establishes an important relationship between catalysts and anomalies. In the present context, a symmetry group is said to be anomalous if no short-range entangled state can be invariant under all elements of the group. A state is said to have short-range entanglement (SRE) if it can be prepared from a product state using an FDQC, and otherwise it has long-range entanglement (LRE) \cite{Chen2010}. In general, we expect that the set of symmetries $\mathcal{U}$ and $U(g)$ for all $g\in G$ is anomalous whenever $\mathcal{U}$ is an entangler for a non-trivial SPT phase. One can argue this based on the fact that these are the symmetries that emerge on the boundary of a decorated domain wall SPT-ordered state in one higher dimension \cite{Chen2014ddw}, and it is known that such states do not allow short-range entangled boundary states. For the case of abelian symmetries in 1D systems, we explicitly prove in Sec.~\ref{sec:mixed} that $(G,\mathcal{U})$-symmetric catalysts must have LRE. The fact that anomalies can ``eat'' SPT phases has previously been observed in Refs.~\cite{Hsin2020,Dumitrescu2024,choi2024noninvertible}.

In 1D, common states with LRE include symmetry-breaking states with cat-like entanglement, ground states of gapless Hamiltonians, and states consisting of long-range entangled pairs of qubits. Once we go to higher dimensions, there are other kinds of states with LRE including topologically ordered states. We will see that all of these classes of states are useful as catalysts for certain SPT phases. In Sec.~\ref{sec:mixed}, we will see that mixed-state catalysts do not need to be long-range entangled. 

\subsection{Case study: 1D SPT and LSM} \label{sec:1dlsm}

\begin{figure}
    \centering
    \includegraphics[width=0.8\linewidth]{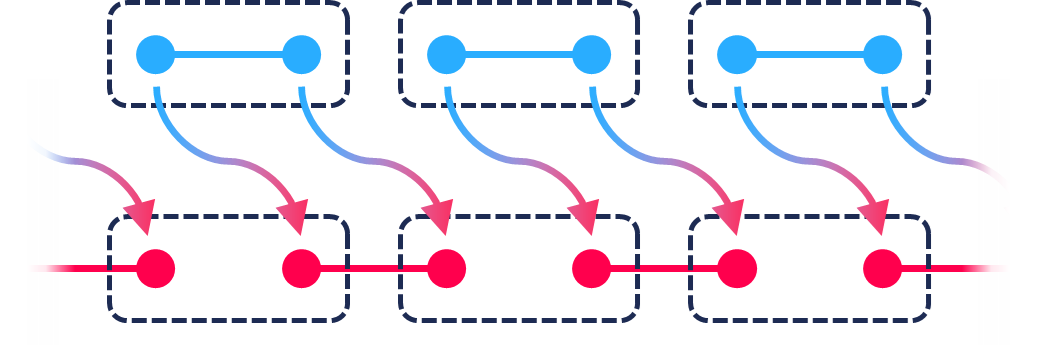}
    \caption{The trivial (top) and non-trivial (bottom) SPT states where each pair of dots connected by a line represents a Bell state. The two-qubit unit cell is indicated by the dotted boxes, and the arrows denote the action of the lattice translation operator $\mathcal{T}$ that relates the two states.}
    \label{fig:1dspt}
\end{figure}

Here, we present a case study of SPT catalysts that highlights some of the different kinds of catalysts that fall under the framework of Theorem~\ref{thm:catalyst}. In this case, the anomaly described above will be the familiar Lieb-Schultz-Mattis (LSM) anomaly \cite{Lieb1961}. The two states we consider are defined on $2N$ qubits with periodic boundary conditions as follows,
\begin{equation} \label{eq:dimer_triv}
    |\triv\rangle = \bigotimes_{i=1}^{N} |\Omega\rangle_{2i,2i+1},
\end{equation}
\begin{equation} \label{eq:dimer_triv}
    |\mathrm{SPT}\rangle = \bigotimes_{i=1}^{N} |\Omega\rangle_{2i-1,2i},
\end{equation}
where $|\Omega\rangle = (|00\rangle + |11\rangle)/\sqrt{2}$ is the Bell state. If we block sites into pairs of qubits, $|\triv\rangle$ describes a product of Bell states contained within unit cells, while $|\mathrm{SPT}\rangle$ describes a product of Bell states between unit cells, see Fig.~\ref{fig:1dspt}. Both states are invariant under a $\Z_2\times \Z_2$ symmetry generated by $\prod_i X_i$ and $\prod_i Z_i$. It is well-known that these are representative fixed-point states for the trivial and non-trivial SPT phases with this symmetry \cite{Chen2011,Schuch2011,Chen2013}. Therefore, there is no symmetric FDQC relating them.

To apply Theorem~\ref{thm:catalyst}, we need to identify the entangler $\mathcal{U}$. In this case, one particularly simple choice is to take $\mathcal{U}$ to be the lattice translation operator $\mathcal{T}$ that maps site $i$ to site $i+1$ as shown in Fig.~\ref{fig:1dspt}. We emphasize that $\mathcal{T}$ translates single qubits, rather than unit cells, so it acts non-trivially within a unit cell. The translation operator is an example of a QCA. Furthermore, it clearly commutes with both symmetry operators so it is a symmetric QCA. The fact that $\mathcal{T}\otimes \mathcal{T}^{-1}$ can be implemented as a symmetric FDQC has been discussed previously \cite{Farrelly2020reviewofquantum,Po2016,Stephen2024klocal}. To find a catalyst, we need a state that is both translationally invariant and symmetric under $\prod_i X_i$ and $\prod_i Z_i$. After finding such a state, we can apply Theorem~\ref{thm:catalyst} to construct a symmetric FDQC that translates $|\triv\rangle$ forwards by one site, transforming it into $|\mathrm{SPT}\rangle$, and translates the catalyst backward by one site, which leaves it unchanged because the catalyst is translationally invariant. This accomplishes the desired catalyzed transformation.

When we search for symmetric, translationally invariant states, we run into the well-known LSM anomaly \cite{Lieb1961}. While $\prod_i X_i$ and $\prod_i Z_i$ commute within unit-cells of size two, they anti-commute within a single site since $XZ=-ZX$. Therefore, a translationally-invariant catalyst must carry a fractional charge, \textit{i.e.} a projective representation, per unit cell. The LSM theorem and its later generalizations state that no short-range entangled state can satisfy this constraint \cite{Lieb1961,Schuch2011}. Typical ways to avoid the LSM theorem come from long-range order originating from spontaneous symmetry breaking of the internal or translation symmetry, or critical correlations found in ground states of gapless Hamiltonians. We will find that each of these cases outlines a generic recipe for constructing catalysts.

First, the catalyst can spontaneously break (part of) the internal $\Z_2\times \Z_2$ symmetry. An example of such a catalyst in the Greenberger-Horne-Zeilinger (GHZ) state,
\begin{equation}
    |a\rangle = |\mathrm{GHZ}\rangle = \frac{1}{\sqrt{2}}(|00\cdots 0\rangle + |11\cdots 1\rangle),
\end{equation}
which is indeed symmetric and translationally invariant. Observe that this state has a long-range order for operators charged under the symmetry $\prod_i X_i$. 

Second, the catalyst can spontaneously break translation symmetry. An example of this is obtained by simply forming the superposition 
\begin{equation}
    |a\rangle = \frac{1}{\mathcal{N}}(|\triv\rangle + |\mathrm{SPT}\rangle),
\end{equation}
for some normalization constant $\mathcal{N}$. This state, which is a superposition of dimer-like coverings, has a structure reminiscent of the famous Majumdar-Ghosh state \cite{majumdar1969next}.

Third, the catalyst can be the ground state of a gapless Hamiltonian. To construct such an example, note that the trivial and SPT states are ground states of the following gapped Hamiltonians,
\begin{align}
    H_{\triv} &= -\sum_i X_{2i}X_{2i+1} - \sum_i Z_{2i}Z_{2i+1}, \\
    H_{\mathrm{SPT}} &= -\sum_i X_{2i-1}X_{2i} - \sum_i Z_{2i-1}Z_{2i}.
\end{align}
Then, the we define the Hamiltonian,
\begin{equation} \label{eq:hxy}
    H_a = H_{\triv} + H_{\mathrm{SPT}},
\end{equation}
which is gapless and has a unique ground state with quasi long-range correlations (i.e. two-point correlations decaying as a power law) \cite{Thorngren_2021}. The ground state of this Hamiltonian will be $(G,\mathcal{U})$-symmetric and therefore a valid catalyst.

Finally, we give an example that does not fit into the usual LSM framework since it is not the ground state of a local Hamiltonian. Namely, we can construct the product of long-range Bell states as,
\begin{equation} \label{eq:lr_bell}
    |a\rangle = \bigotimes_{i=1}^N |\Omega\rangle_{i,i+N}.
\end{equation}
This example is interesting because it leads to a novel method to prepare SPT phases with long-range interactions, as discussed in Sec.~\ref{sec:state_prep}. \diff{Although this example does not fit into the LSM framework, it is constrained by a more general condition that we describe in Sec.~\ref{sec:mixed} (Theorem \ref{thm:loc}).}

\section{Catalysts for SPT phases with global symmetries} \label{sec:global_spt}

Here we construct catalysts for all SPT phases that lie within the group cohomology classification \cite{Chen2013}. These phases can be represented by fixed-point states called \textit{cocycle states}, which are defined on a triangulation of a $d$-dimensional space with onsite the Hilbert spaces $\mathbb{C}^{|G|}$ spanned by the states $|g\rangle$ for $g\in G$. The global symmetry is defined by $U(g)=u(g)^{\otimes N}$ where $u(g)|h\rangle = |gh\rangle$. The states in $d$-dimensions are defined by a $d+1$-cocycle which is a function $\nu_{d+1}:G^{d+2} \mapsto U(1)$ satisfying certain conditions, see Ref.~\cite{Chen2013} for more details. Cocycles can be divided into certain equivalence classes called cohomology classes which are in a one-to-one correspondence with the ``in-cohomology'' SPT phases \cite{Chen2013}. These cover all SPT phases in low dimensions \cite{Else2014}, but certain ``beyond-cohomology'' phases can exist in four dimensions and higher \cite{kapustin2014,Fidkowski2020} (or three dimensions if time-reversal symmetry is allowed \cite{Burnell2014}).

Given a cocycle $\nu_d$, we can define a finite-depth quantum circuit $\mathcal{U}(\nu_d)$ such that,
\begin{equation} \label{eq:cocycle_state}
    |\psi(\nu_d)\rangle = \mathcal{U}(\nu_d)|+_G+_G\cdots +_G\rangle,
\end{equation}
is a fixed-point state belonging to the corresponding SPT phase, where $|+_G\rangle = \frac{1}{\sqrt{|G|}}\sum_{g\in G}|g\rangle$ \cite{Chen2013}. The circuit $\mathcal{U}(\nu_d)$ commutes with the symmetry $U(g)$ as a whole. However, it is only a symmetric FDQC (meaning the individual gates are symmetric) if $\nu_d$ belongs to a trivial cohomology class. Without loss of generality, we can always assume that $\mathcal{U}(\nu_d)$ has finite order, meaning that $\mathcal{U}(\nu_d)^L=I$ for some $L$. This is physically related to the fact that the set of SPT phases with symmetry $G$ forms a finite group, see Appendix~\ref{sec:cocycle_details} for details.

To constructs catalysts for cocycle states, we follow Theorem~\ref{thm:catalyst} and search for $(G,\mathcal{U}(\nu_d))$-symmetric states.
The first approach is to construct a state that spontaneously breaks the internal symmetry $G$,
\begin{equation} \label{eq:cocycle_catalyst_1}
    |a\rangle = \frac{1}{\sqrt{|G|}}\sum_{g\in G} |gg\dots g\rangle.
\end{equation}
This state is $G$-symmetric by construction. In Appendix~\ref{sec:cocycle_details}, we show that $\nu_d$ can again be chosen such that $\mathcal{U}(\nu_d)|gg\dots g\rangle = |gg\dots g\rangle$ for all $g\in G$, so the state is also $\mathcal{U}(\nu_d)$-symmetric. This state completely breaks the $G$ symmetry. More generally, it is sufficient to only break the symmetry down to a subgroup $H$ which has the property that the SPT phase is trivial when only $H$ symmetry is enforced.

Alternatively, we can create a state that spontaneously breaks $\mathcal{U}(\nu_d)$ symmetry,
\begin{equation} \label{eq:cocycle_catalyst_2}
    |a\rangle = \frac{1}{\mathcal{N}} \sum_{\ell=1}^{L} \mathcal{U}(\nu_d)^\ell |+_G+_G\cdots +_G\rangle,
\end{equation}
for a normalization constant $\mathcal{N}$. This state is $(G,\mathcal{U}(\nu_d))$-symmetric by construction.

The final construction uses a Hamiltonian perspective, and often produces gapless catalysts (\textit{i.e.} states that are the ground states of gapless Hamiltonians). Consider the trivial paramagnetic Hamiltonian $H_0$ which has $|+_G+_G\cdots +_G\rangle$ as its unique gapped ground state. Then we construct the following Hamiltonian,
\begin{equation} \label{eq:ham_cat}
    H_a = \sum_{\ell=1}^L \mathcal{U}(\nu_d)^\ell H_{\triv} \mathcal{U}(\nu_d)^{\ell\dagger}.
\end{equation}
Since this Hamiltonian commutes with $G$ and $\mathcal{U}$ by construction, its ground state will be $(G,\mathcal{U})$-symmetric. If there are multiple ground states, we can always find a symmetric state within the ground space. For the special case where $\ell = 2$, Eq.~\eqref{eq:ham_cat} can be interpreted as being part of a continuous family of Hamiltonians,
\begin{equation}
    H_a(\alpha) = \alpha H_{\triv} + (1-\alpha) H_{\spt},
\end{equation}
where we defined $H_{\spt}=\mathcal{U}(\nu_d) H_{\triv} \mathcal{U}(\nu_d)^{\dagger}$ and $H_a\equiv H_a(1/2)$. This Hamiltonian is self-dual since the two terms are exchanged after conjugation by $\mathcal{U}(\nu_d)$. 
In some cases, $\alpha=1/2$ corresponds to a second-order phase transition such that $H_a$ is gapless and has a unique ground state with critical correlations, as in Eq.~\eqref{eq:hxy}. However, this is not always the case, and $\alpha=1/2$ may instead correspond to a first-order phase transition or belong to an intermediate symmetry-breaking phase, as discussed in Refs.~\cite{Tsui2015, Bultinck2019, Tantivasadakarn2023building, Tantivasadakarn2023pivot}.

Finally, we remark that Theorem~\ref{thm:catalyst} can be applied to certain beyond-cohomology phases as well. For example, the 4D lattice model described in Ref.~\cite{Fidkowski2020} is a beyond-cohomology SPT, but it was shown that there still exists a symmetric QCA $\mathcal{U}$ that creates it from a symmetric product state. Therefore, we can still search for $(G,\mathcal{U})$-symmetric states to catalyse this phase.

\subsection{Example: 1D cluster state} \label{sec:cluster}

As the simplest example of a cocycle state, we can consider the 1D cluster state \cite{Raussendorf2003}. Strictly speaking, the cluster state is not technically a cocycle state, but it can be straightforwardly re-interpreted as one, as discussed in Ref.~\cite{Miller2018}. For accessibility of the discussion, we will use the standard presentation of the cluster state. The state is defined on a ring of even length $N$ as,
\begin{equation}
    |\psi_C\rangle = \mathcal{U}_\CZ |++\cdots +\rangle,
\end{equation}
where $|+\rangle = (|0\rangle + |1\rangle)/\sqrt{2}$ 
 and the entangler $\mathcal{U}_\CZ$ is defined as,
\begin{equation} \label{eq:1dcluster_entangler}
    \mathcal{U}_\CZ = \prod_{i=1}^N \CZ_{i,i+1},
\end{equation}
where the two-qubit controlled-$Z$ operator is defined as $\CZ_{i,j} = |0\rangle\langle 0|_i\otimes I_j + |1\rangle\langle 1|_i \otimes Z_j$. The state has a $\\Z_2\times \Z_2$ symmetry generated by $U_e = \prod_{i} X_{2i}$ and $U_o = \prod_{i} X_{2i+1}$. 

Let us explore the various kinds of catalysts constructed in the previous section for the case of the cluster state. The symmetry-breaking catalyst defined in Eq.~\eqref{eq:cocycle_catalyst_1} corresponds to a pair of GHZ states on the even and odd sublattices. In this case, the SPT order can be trivialized by spontaneously breaking the $\Z_2\times \Z_2$ symmetry to any $\Z_2$ subgroup. Accordingly, a state made of a GHZ state on one sublattice and a trivial product state on the other is also a valid $(G,\mathcal{U}_{\CZ})$-symmetric catalyst.
Next, since $\mathcal{U}_\CZ^2 = \id$, the catalyst associated to spontaneously breaking the entangler symmetry is $(|++\cdots +\rangle + |\psi_C\rangle)/\mathcal{N}$. This state has recently received attention since it is dual to a state with non-invertible symmetry \cite{Seifnashri2024}.
Finally, the family of Hamiltonians,
\begin{equation}
    H(\alpha) = -\alpha\sum_i X_i -(1-\alpha)\sum_i Z_{i-1}X_iZ_{i+1},
\end{equation}
commutes with $\mathcal{U}_\CZ$ at the $\alpha=1/2$ point. At this point, the Hamiltonian is gapless \cite{Levin2012}, and the ground state is critical with central charge $c=1$. 

\subsection{Topologically ordered catalysts in $d>1$} \label{sec:2dspt}

So far, we have primarily studied catalysts with long-range or algebraically decaying correlations. This is because these are the kinds of states that can be symmetric under an anomalous symmetry (``match the anomaly'') in 1D. In higher dimensions, another type of catalyst becomes possible, since it is known that in higher dimensions anomalies can be matched using topologically ordered states with symmetry fractionalization \cite{Chen2015anomalous}. Accordingly, we will find that such states can also be catalysts for higher-dimensional SPT phases with global symmetries. We first argue this from the viewpoint of the algebraic theory of anyons, and then we give an explicit lattice example.

\subsubsection*{Algebraic argument}

To demonstrate the algebraic argument, let us first explain how one could argue the equivalence of two anyon models, which are characterized by a set of anyons having certain fusion rules and braiding statistics. As an example, consider the equivalence between a stack of two double semion (DS)) models and a stack of one toric code (TC) model and one DS model. While these two systems initially appear to have different anyon content, it turns out that, by relabelling the anyons in a way that preserves the fusion structure, the two systems can be made equivalent \cite{Shirley2019fractional}. This implies that the two systems are equivalent as topological phases, and therefore suggests the presence of a FDQC relating them. Such an FDQC was explicitly constructed in Ref.~\cite{Shirley2019fractional}.

To establish the corresponding argument for SPT phases, one can try to ``ungauge'' the topological order in the first layer of the stacks, which maps the TC and DS phases to the trivial and non-trivial SPT phases with $\Z_2$ symmetry, respectively \cite{Levin2012}. To discuss the algebraic perspective in the presence of symmetry, we enter the regime of symmetry-enriched topological (SET) phases, where we need to consider the joint algebraic theory containing both the anyons and the extrinsic symmetry defects \cite{Barkeshli2019}. When we compare a DS model stacked with a trivial $\Z_2$ paramagnet to a DS model stacked with a $\Z_2$ SPT, the only difference is that the extrinsic $\Z_2$ symmetry defect has trivial self-statistics in the former case while it has semionic statistics in the latter case (since this defect becomes the semion after gauging) \cite{Levin2012}. However, we can fix this by binding the semion from the DS model to the symmetry flux in one model, akin to relabelling anyons above. This binding can again be implemented using an FDQC and, since it does not change how the symmetry defects interact with the symmetry charges, one can choose this FDQC to be symmetric. Therefore, the argument suggests that the DS model can serve as a catalyst for the $\Z_2$ SPT order. 

In general, one can apply the above algebraic perspective to search for SET phases that can ``absorb'' a given SPT phase.
While the above SET models have trivial symmetry fractionalization, meaning that the $\Z_2$ symmetry acts trivially on the anyons, more general SET catalysts may require non-trivial symmetry fractionalization. 
The fact that SET phases can ``absorb'' SPT phases, or equivalently the fact that SET phases can have gapped interfaces with SPT phases, has also previously been discussed in Refs.~\cite{Wang2013weak, Lu2014gapped, Cheng2020relative}. Here, we are reinterpreting this fact from the perspective of many-body catalysts. 

\subsubsection*{Explicit catalyst for $\Z_2$ SPT}

We can confirm the validity of the algebraic argument above using an explicit lattice model. We consider the modified form of the DS model presented in Ref.~\cite{Dauphinais2019} (see also \cite{vonKeyserlingk2013, MagdalenadelaFuente2021nonpauli}). This model is defined on a triangular lattice with qubits on the edges $e\in E$, and has the following Hamiltonian,
\begin{equation}
    H_{DS} = -\sum_p B_p -\sum_v A_v,
\end{equation}
where $B_p = \prod_{e \in p} Z_e$ is the product of $Z$'s around a triangular plaquette, $A_v = \beta_v \prod_{e \ni v} X_e$ is a product of $X$'s on all edges emanating from the vertex $v$ times a certain diagonal unitary $\beta_v$ that acts on the 12 edges nearest the vertex $v$. The precise form of $\beta_v$ is not important here, but we note the properties (i) $[A_v,A_{v'}]=0$ for all $v,v'$, (ii) $A_v^2=I$, and (iii) $\prod_v A_v = I$ (on periodic boundaries) \cite{Dauphinais2019}. Importantly, this model has the same topological order as the original DS model. To equip this model with a $\Z_2$ symmetry, we add decoupled qubits on the vertices $v\in V$ of the lattice and define the $\Z_2$ symmetry $U_{Z_2}=\prod_v X_v$ and the modified symmetric Hamiltonian $\widetilde{H}_{DS} = H_{DS} - \sum_v X_v$, such that $\widetilde{H}_{DS}$ has ground states of the form $|\widetilde{DS}\rangle = |DS\rangle \otimes |+\rangle^{\otimes |V|}$ where $|DS\rangle$ is a ground state of $H_{DS}$ and we introduced the notation $|\psi\rangle^{\otimes |\Gamma|} = \bigotimes_{i\in \Gamma} |\psi\rangle_i$.

The $\Z_2$ SPT we will catalyze also takes an unconventional form. Inspired by the construction in Ref.~\cite{Webster2022xpstabiliser}, we consider the model,
\begin{equation}
    H_{\spt} = -\sum_p B_p - \sum_v A_v X_v - \sum_e Z_eZ_{v_1}Z_{v_2},
\end{equation}
where $v_{1,2}$ are the two vertices connected to edge $e$. This model again commutes with $U_{Z_2}$. To confirm that this model belongs to the non-trivial $\Z_2$ SPT phase, observe that gauging the $\Z_2$ symmetry results in a model equivalent to $H_{DS}$ up to an FDQC, which implies the non-trivial SPT order \cite{Levin2012}. The ground state of $H_{\spt}$ can be written in the following way,
\begin{equation}
    |\spt\rangle = \mathcal{U}_{\spt} \left(|0\rangle^{\otimes |E|} \otimes |+\rangle^{\otimes |V|}\right),
\end{equation}
where we define the SPT entangler,
\begin{equation}
    \mathcal{U}_{\spt} = \prod_{v} \mathrm{ctrl}-A_v,
\end{equation}
where the operator $\mathrm{ctrl}-A_v$ is defined to act on the vertex $v$ and the twelve edges closest to $v$ as $\mathrm{ctrl}-A_v = |0\rangle\langle 0|_v\otimes I + |1\rangle\langle 1|_v\otimes A_v$. 
To confirm that $|\spt\rangle$ is indeed the ground state of $H_{\spt}$, we note that $\mathcal{U}_{\spt} X_v \mathcal{U}_{\spt}^\dagger = X_v A_v$ (owing to the fact that $A_v^2 = 1$) which leads to, 
\begin{equation}
    H_{\triv} := \mathcal{U}_{\spt} H_{\spt} \mathcal{U}_{\spt}^\dagger = -\sum_p B_p - \sum_v X_v - \sum_e Z_e.
\end{equation}
This Hamiltonian has a unique ground state $|\triv\rangle = |0\rangle^{\otimes |E|} \otimes |+\rangle^{\otimes |V|}$, so $H_{\spt}$ must have the unique ground state $|\spt\rangle = \mathcal{U}_{\spt}|\triv\rangle$. The SPT entangler $\mathcal{U}_{\spt}$ can be written as an FDQC because the different $\mathrm{ctrl}-A_v$ gates all commute (owing to the fact that the $A_v$ all commute), and $U_{Z_2}\mathcal{U}_{\spt}U_{Z_2}^\dagger = \mathcal{U}_{\spt}\prod_v A_v = \mathcal{U}_{\spt}$ so it is a symmetric QCA. 

Finally, we show that the ground states $|\widetilde{DS}\rangle$ of $\widetilde{H}_{DS}$ are $(G,\mathcal{U}_{\spt})$-symmetric for $G=Z_2$, and therefore serve as catalysts for $|\spt\rangle$ according to Theorem \ref{thm:catalyst}. For any ground state $|\widetilde{DS}\rangle$, we have,
\begin{equation}
    \mathcal{U}_{\spt}|\widetilde{DS}\rangle= \sum_{\vec{i}} \left(\prod_{v} A_v^{i_v}\right) |DS\rangle \otimes |\vec{i}\rangle = |\widetilde{DS}\rangle,
\end{equation}
where $|\vec{i}\rangle = \bigotimes_v |i_v\rangle$ with $i_v=0,1$ and we used the fact that $A_v|DS\rangle = |DS\rangle$ for all $v$ since $|DS\rangle$ is a ground state. Therefore, $|\widetilde{DS}\rangle$, interpreted as a SET model with DS topological order and a trivial $\Z_2$ symmetry fractionalization, is a catalyst for the non-trivial $\Z_2$ SPT.

\section{Catalysts for other types of symmetries} \label{sec:other_symms}

In all examples until this point, we considered only global onsite symmetries of the form $U(g)=u(g)^{\otimes N}$. However, Theorem~\ref{thm:catalyst} applies equally well to cases where the action of the symmetries is not uniform across space. This occurs in many cases of physical interest, particularly higher-form symmetries and subsystem symmetries. 
In this section, we show what new kinds of catalysts emerge when considering these more general symmetries, and what this teaches us about the corresponding phases of matter.

\subsection{Higher-form symmetries} \label{sec:higherform}

A $q$-form symmetry operator is defined as one which is supported on a codimension-$q$ subset of the lattice, such that a $0$-form symmetry is a global symmetry. 
As an example with higher-form symmetries, we consider the 2D cluster state with qubits on the edges $e\in E$ and vertices $v\in V$ of a square lattice (the ``Lieb lattice'') \cite{Yoshida2016,verresen2024higgs}. The parent Hamiltonian of the 2D cluster state is,
\begin{equation}
H_{\mathrm{Lieb}}=-\sum_v\left(X_v\prod_{e\ni v}Z_e\right)-\sum_e\left(X_e\prod_{v\in e}Z_v\right),
\end{equation}
where subscripts $e$ and $v$ depict edge and vertex, respectively. All terms in the Hamiltonian commute with each other, and the 2D cluster state is exactly the ground state. The ground state can be written as,
\begin{equation}
    |\psi_{\mathrm{Lieb}}\rangle = \mathcal{U}_{\mathrm{Lieb}}\left(|+\rangle^{\otimes |E|} \otimes |+\rangle^{\otimes |V|}\right),
\end{equation}
where the entangler is defined as,
\begin{equation}
    \mathcal{U}_{\mathrm{Lieb}} = \prod_{\langle e,v\rangle} CZ_{e,v}.
\end{equation}
This state has both a 0-form $\Z_2^{(0)}$ symmetry and a 1-form $\Z_2^{(1)}$ symmetry, with the following symmetry generators:
\begin{align}
U^{(0)}=\prod_vX_v,\quad U^{(1)}_\gamma=\prod_{e\in\gamma}X_e,
\end{align}
where $\gamma$ is an arbitrary closed loop along the bonds of the square lattice. 

To construct catalysts for this state, we focus on symmetry-breaking states. Depending on which symmetry we spontaneously break, these states will look very different. For example, the following GHZ state breaks $\Z_2^{(0)}$ symmetry,
\begin{equation}
    |a\rangle = |+\rangle^{\otimes |E|} \otimes \left(|0\rangle^{\otimes |V|} + |1\rangle^{\otimes |V|} \right),
\end{equation}
and one can straightforwardly confirm that is is invariant under $\mathcal{U}_{\mathrm{Lieb}}$ and is therefore a valid catalyst. We can instead spontaneously break the $\Z_2^{(1)}$ symmetry. Recalling that a spontaneously broken 1-form symmetry is synonymous with topological order \cite{Nussinov2009}, we can define the state,
\begin{equation}
    |a\rangle = \left( \prod_p \frac{1 + B_p}{2} |0\rangle^{\otimes |E|} \right)\otimes |+\rangle^{\otimes |V|},
\end{equation}
where $B_p = \prod_{e\in p} X_e$ is a product of four $X$ operators around the plaquette $p$. The state on the edge degrees of freedom is precisely the ground state of the 2D toric code. This state is also invariant under $\mathcal{U}_{\mathrm{Lieb}}$ and is therefore a valid catalyst. 
The fact that we can catalyze $|\psi_{\mathrm{Lieb}}\rangle$ by spontaneously breaking only one of the $\Z_2^{(0)}$ or $\Z_2^{(1)}$ symmetries implies that it should belong to an trivial SPT phase when only one of these symmetries is enforced. This can be confirmed by constructing explicit symmetric FDQCs. Thus, we find that the catalyst perspective can be potentially useful in determining when an SPT is trivial or not.

For higher-form symmetries, topologically ordered phases can serve as catalysts due to their interpretation as states with spontaneously broken higher-form symmetry. This is distinct from Sec.~\ref{sec:2dspt} where the presence of symmetry-fractionalized excitations in topological phases is what enabled their use as catalysts. 

The above example is protected by a combination of $0$-form and $1$-form symmetries, but we can similarly analyze models that have only higher-form symmetries. For example, the 3D cluster state defined in Ref.~\cite{Raussendorf2005} has $\Z_2^{(1)}\times \Z_2^{(1)}$ symmetry \cite{Roberts2017} and can be catalyzed using ground states of 3D toric codes \cite{Hamma2005}.

\subsection{Subsystem symmetries} \label{sec:subsystem}

Subsystem symmetries are similar to higher-form symmetries in that they involve symmetry operators that act on lower-dimensional subsets of the lattice. They are however distinguished from higher-form symmetries by the fact that these subsets are rigid, i.e. non-deformable.  Here, we consider line-like symmetries, but similar constructions can be made for systems \cite{VHF} with fractal or planar symmetries as well. 

A SPT phase that is protected by subsystem symmetries is called a subsystem SPT (SSPT) phase of matter \cite{You2018}. As the simplest example, we again consider a 2D cluster state, but this time with qubits only living on the vertices $v$ of a square lattice. The parent Hamiltonian is,
\begin{equation}
H_{\mathrm{square}}=-\sum_v X_v \left(\prod_{v'\in n(v)}Z_{v'}\right),
\end{equation}
where $n(v)$ contains all neighbouring vertices of $v$. The ground state is,
\begin{equation}
    |\psi_{\mathrm{square}}\rangle = \mathcal{U}_{\mathrm{square}}|+\rangle^{|V|},
\end{equation}
where,
\begin{equation}
    \mathcal{U}_{\mathrm{square}} = \prod_{\langle v,v'\rangle} \CZ_{v,v'}.
\end{equation}
The subsystem symmetries of this state act along diagonal lines in both directions,
\begin{equation}
    U_{c,\pm} = \prod_i X_{i,c\pm i},
\end{equation}
where each value of $c$ describes a different vertical line and $\pm$ denote the two directions (diagonal and anti-diagonal).

This cluster state belongs to a non-trivial subsystem SPT (SSPT) phase of matter \cite{You2018}. In the classification of SSPT phases with line-like symmetries, there is a distinction between strong and weak phases, where the latter is equivalent to a stack of 1D SPT phases aligned with the direction of the symmetries and the former is not. From this definition, it is clear that a suitable catalyst for a weak SSPT can be constructed using a stack of 1D states, each being able to catalyze a 1D SPT phase. Conversely, we conjecture that a strong SSPT cannot be catalyzed by a stack of 1D states. This would give another way to distinguish strong and weak SSPT phases, adding to the approaches discussed in Refs.~\cite{You2018,Devakul2018strong}. 

The justification for this conjecture is the following. As we have seen, and as is emphasized by Theorem~\ref{thm:catalyst}, catalysts allow one to act with symmetric QCA on the target system. For this argument, we assume this is essentially all that we can accomplish with catalysts. Now, suppose we embed a 1D catalyst along some line of the 2D SSPT. By acting with a symmetric FDQC, we expect that this allows us to modify our system by a 1D symmetric QCA supported along the region near the catalyst. By stacking many 1D catalysts, we are able to act with many 1D symmetric QCA on different 1D lines spanning the 2D state. If we stack a constant number of such catalysts per generator of the subsystem symmetry, we can overall apply a finite-depth circuit of 1D symmetric QCA. This kind of transformation is exactly the basis of the classification of strong SSPTs given in Ref.~\cite{Devakul2018strong}. In particular, it was argued that a strong SSPT cannot be created from a symmetric product state using such a transformation. Therefore, we conjecture that a stack of 1D catalysts cannot catalyze a strong 2D SSPT.

For an example of a symmetry-breaking state that \textit{can} be a catalyst for $|\psi_{\mathrm{square}}\rangle$, we have the symmetric ground state of the plaquette Ising model,
\begin{equation}
    H_{PIM} = -\sum_v \left(\prod_{v'\in n(v)}Z_{v'}\right).
\end{equation}
This actually defines two plaquette Ising models, one for each sublattice. This is an intrinsically 2D symmetry-breaking model. The symmetric ground state of $H_{PIM}$ that is invariant under all $U_{c,\pm}$ is a subsystem symmetry-breaking state that can catalyze $|\psi_{\mathrm{square}}\rangle$.

\section{Mixed state catalysts} \label{sec:mixed}

In this section, we go beyond pure state catalysts and consider general mixed states. We will see that this allows us to define catalysts which much different properties that pure state catalysts. In particular, mixed state catalysts need not have long-range correlations. When using mixed states to catalyze transformations between pure states, we restrict to strongly-symmetric catalysts which satisfy,
\begin{equation} \label{eq:strong_symm}
    U(g)\rho = \rho U(g) = \rho, \quad \forall g\in G.
\end{equation}
This means that the state $\rho$ is a mixture of symmetric states, rather than only being symmetric ``on average'', such that the symmetry of the system is always preserved. If we allow weakly-symmetric states which don't satisfy Eq.~\eqref{eq:strong_symm} but do satisfy $U(g)\rho U(g)^\dagger = \rho$, then it is trivial to construct catalysts, the maximally mixed state being one example. \diff{This matches the known fact that SPT order is not stable under weakly-symmetric channels \cite{de_Groot_2022}.}
We then define a mixed state catalyst in a similar fashion as for pure states,
\begin{definition}[Mixed state catalysts] \label{def:catalyst}
    Given a pair of $G$-symmetric states $|\psi_0\rangle$ and $|\psi_1\rangle$, we say that a mixed state $\rho$ \emph{catalyzes} the transformation $|\psi_0\rangle \mapsto |\psi_1\rangle$ if $\rho$ is strongly $G$-symmetric and there exists a $G$-symmetric FDQC $\mathcal{V}$ such that,
    \begin{equation}
        \mathcal{V}(|\psi_0\rangle\langle \psi_0| \otimes \rho)\mathcal{V}^\dagger  = |\psi_1\rangle\langle \psi_1| \otimes \rho.
    \end{equation}
\end{definition}
Note that we still only drive the catalyzed transformation using an FDQC. We could further generalize this by asking that $|\psi_0\rangle\langle \psi_0| \otimes \rho$ and $|\psi_1\rangle\langle \psi_1| \otimes \rho$ are two-way connected via symmetric local channels, akin to recent definitions of SPT phases for mixed-states \cite{Average_2023},  but we will find that our definition is sufficient to allow for interesting examples of mixed state catalysts.  

The notion of a $(G,\mathcal{U})$-symmetric catalyst still holds for mixed states. If $\rho$ is weakly symmetric under $\mathcal{U}$, meaning that $\mathcal{U}\rho\mathcal{U}^\dagger = \rho$, then a straightforward generalization of the proof of Theorem \ref{thm:catalyst} shows that it can catalyze the transformation $|\psi\rangle \mapsto \mathcal{U}|\psi\rangle$ for all $|\psi\rangle$. We therefore still use the terminology $(G,\mathcal{U})$-symmetric catalysts for mixed states, with the understanding that the state should be strongly $G$-symmetric but need only be weakly $\mathcal{U}$-symmetric. 

\subsection{Catalysts from strong-to-weak symmetry breaking}

An interesting class of mixed state catalysts can be constructed using the notion of strong-to-weak SSB (SW-SSB) \cite{Lee_2023, Average_2023, SWSSB2024, sala2024spontaneous, huang2024hydro, xu2024average, moharramipour2024symmetry, gu2024spontaneous, guo2024new, Kuno_2024, sala2024decoherence, zhang2024FDT}. We showcase this for the example of the 1D cluster state discussed in Sec.~\ref{sec:cluster}. Consider the following mixed state,
\begin{equation} \label{eq:swssb_cat}
    \rho = \frac{1}{2^N}(\id + U_e)(\id + U_o).
\end{equation}
This state is clearly strongly $G$-symmetric for $G=\Z_2\times \Z_2$. It is also weakly $\mathcal{U}_\CZ$-symmetric since,
\begin{align}
&2^NU_{\CZ}\rho U_{\CZ}\nonumber\\
&=\left(\id+\prod_jZ_{2j-1}X_{2j}Z_{2j+1}\right)\left(\id+\prod_jZ_{2j}X_{2j+1}Z_{2j+2}\right)\nonumber\\
&=\left(\id+U_e\right)\left(\id+U_o\right).
\end{align}
Therefore, it is a valid catalyst.  In contrast with pure state catalysts in 1D which always had (quasi) long-range correlations, this state only has short-ranged two-point correlations of the form $\langle O_i O_j\rangle = \tr(\rho O_iO_j)$. This follows from the fact that the state $\rho$ can also be written as a convex sum of product states, namely a sum of all rank-one projectors of the form $|+-+-+\cdots\rangle\langle +-+-+\cdots|$ containing an even number of $|-\rangle$ states on both the even and odd sublattices.
Because of this, it is also infinite-partite separable and therefore has no long-range entanglement \cite{lessa2024anomaly}. In fact, it can be prepared from a symmetric product state with a symmetric local quantum channel (although the reverse transformation is not possible), as discussed further in Sec.~\ref{sec:state_prep} in the context of efficient preparation of catalysts. Thus, its properties starkly contrast with those of pure state catalysts. 

The general framework to which this state belongs is that of strong-to-weak SSB (SW-SSB) \cite{SWSSB2024}. This describes the scenario where a density matrix is strongly $G$-symmetric as a whole, but it is a mixture of a finite number of weakly $G$-symmetric states. More precisely, an SW-SSB mixed state is characterized by a long-ranged fidelity correlator of the local charged operator $O_i$, namely 
\begin{equation} \label{eq:fidelity_corr}
F\left(\rho,O_iO_j^\dag\rho O_jO_i^\dag\right)\sim O(1),
\end{equation}
where $F(\rho,\sigma)=\tr\sqrt{\sqrt{\rho}\sigma\sqrt{\rho}}$ is the fidelity of two density matrices $\rho$ and $\sigma$. This is the analog of the long-range correlator $\tr(\rho O_iO_j)$ that characterizes conventional SSB. One can easily verify that the state $\rho$ defined in Eq.~\eqref{eq:swssb_cat} has $F\left(\rho,Z_iZ_j^\dag\rho Z_jZ_i^\dag\right)=1$ whenever $i$ and $j$ belong to the same sublattice, so it indeed has SW-SSB.

We see that, when mixed-state catalysts are allowed, (quasi) long-range correlations are not necessary, and long-range fidelity correlators are sufficient. The fidelity correlator is closely related to the Edwards-Anderson correlator \cite{SWSSB2024}. Physically, this implies that we can construct catalysts that spontaneously break the symmetry in the sense of a spin-glass does rather than a ferromagnet.

The generalization of the state in Eq.~\eqref{eq:swssb_cat} to general symmetry groups is,
\begin{equation}
    \rho \propto \sum_{g\in G} U(g),
\end{equation}
which will always be $(G,\mathcal{U})$-symmetric for any $G$-symmetric QCA $\mathcal{U}$. This is therefore a valid SW-SSB catalyst for all examples discussed prior to this point. 

As a particularly interesting example, consider the higher-form SPT $|\psi_{\mathrm{Lieb}}\rangle$ discussed in Sec.~\ref{sec:higherform}. We can define the density matrix,
\begin{equation}
    \rho=\left(\frac{\id + \prod_v X_v}{2}\right)\prod_p\left(\frac{\id+B_p}{2}\right).
\end{equation}
This state is invariant under the entangler $\mathcal{U}_{\mathrm{Lieb}}$. It also has SW-SSB of the $\Z_2^{(0)}\times\Z_2^{(1)}$ symmetry. For qubits on vertices, there is SW-SSB of $\Z_2^{(0)}$ captured by a long-ranged fidelity correlator,
\begin{align}
F\left(\rho, Z_vZ_{v'}\rho Z_vZ_{v'}\right)=1.
\end{align}
For the 1-form symmetry generated by $U^{(1)}_\gamma$, the relevant dual operators whose expectation values diagnose symmetry breaking are products of $Z$ operators $W_\lambda = \prod_{e\in\lambda} Z_e$ where $\lambda$ is an closed path on the dual lattice. We can easily check that $\rho$ has SW-SSB of the 1-form $\Z_2^{(1)}$ symmetry in the sense of a vanishing expectation value and a non-vanishing fidelity expectation value of $W_\lambda$,
\begin{align}
\begin{gathered}
\tr\left(\rho W_\lambda\right)=0\\
F\left(\rho, W_\lambda \rho W_\lambda\right)=1.
\end{gathered}
\end{align}
Furthermore, we can check the disorder parameters defined by truncating the 1-form symmetry operators $U^{(1)}_\gamma$ to open strings $\bar{\gamma}$,
\begin{align}
\begin{gathered}
\tr\left(\rho U^{(1)}_{\bar{\gamma}}\right)=0\\
F\left(\rho, U^{(1)}_{\bar{\gamma}} \rho U^{(1)}_{\bar{\gamma}}\right)=1.
\end{gathered}
\end{align}
This pattern of expectation values matches the recent characterization of 1-form SW-SSB given in Ref.~\cite{zhang2024strong}. We remark that, in a state with SW-SSB of a 1-form symmetry, the edge qubits form a \textit{mixed-state topological order} that can restore a classical bit \cite{Dennis_2002, Lee_2023, bao2023mixed, sohal2024noisy, ellison2024classification, zhang2024strong}. 

\subsection{Characterizing 1D catalysts}

We have seen that catalysts in spatial dimension $d>1$ can belong to various different classes with different properties including topological order. In $d=1$, however, the situation is relatively simple, and here we derive strong constraints on possible catalysts in 1D. We focus on $(G,\mathcal{U})$-symmetric catalysts, but we believe that these findings should hold more generally. 

Our constraints are phrased in terms of the concept of symmetry localization. Recall that we are considering mixed-state catalysts in the general case, and a $(G,\mathcal{U})$-symmetric catalyst is then a mixed state $\rho$ such that $U(g)\rho = \rho U(g) = \rho$ and $\mathcal{U}\rho\mathcal{U}^\dagger = \rho$. We can then define two notions of symmetry localization. First, given any finite region $\Gamma = [a,b]$ of the lattice, define $U_{\Gamma,g} = \prod_{i\in \Gamma} u_i(g)$ to be the symmetry truncated to $R$. Then, \textit{strong localization} of the symmetry means that there exist operators $L_{\Gamma,g}$ and  $R_{\Gamma,g}$ that supported in a finite region around sites $a$ and $b$ such that the following holds (perhaps up to exponentially small corrections in $|\Gamma|$),
\begin{equation} \label{eq:strong_loc}
     \rho U_{\Gamma,g}  =  \rho  (L_{\Gamma,g}\otimes R_{\Gamma,g}).
\end{equation}
In the context of SPT phases, the operator $U_{\Gamma,g} (L_{\Gamma,g}^\dagger \otimes R_{\Gamma,g}^\dagger)$  is called a string-order parameter \cite{Pollmann2012}. 

\begin{Theorem} \label{thm:loc}
    Suppose $\rho$ is a 1D state that is strongly $G$-symmetric for a finite abelian group $G$ and satisfies $\mathcal{U}\rho\mathcal{U}^\dagger=\rho$ where $\mathcal{U}$ is an entangler for a non-trivial SPT phase. Then $\rho$ cannot exhibit strong localization of the $G$ symmetry.
\end{Theorem}
\begin{proof}
    We prove this using the invariant for SPT entanglers defined in Ref.~\cite{Zhang2023topological}. We use this approach because these invariants are defined for non-abelian and higher-dimensional SPT phases as well, meaning it is likely that this proof can be extended to those cases as well (as opposed to a proof based on string-order parameters, say). For abelian $G$, the invariant of Ref.~\cite{Zhang2023topological} is defined via the identity,
    \begin{equation} \label{eq:invariant}
        \mathcal{U}^\dagger U_{A,g}\mathcal{U}U_{B,h}\mathcal{U}^\dagger U_{A,g}^\dagger\mathcal{U}U_{B,h}^\dagger = c_{g,h}\id,
    \end{equation}
    where $A = [a,c]$ and $B=[b,d]$ are two overlapping intervals with $a\ll b\ll c\ll d$. The set of phases $c_{g,h}\in U(1)$ for $g,h\in G$ completely specify $\mathcal{U}$ up to symmetric FDQCs. In particular, $\mathcal{U}$ cannot entangle a non-trivial SPT if $c_{g,h}=1$ for all $g,h\in G$. 
    
    We prove that $c_{g,h}=1$ whenever there is a state $\rho$ that exhibits strong symmetry localization and satisfies $\rho=\mathcal{U}\rho\mathcal{U}^\dagger$.
    Assume that such a state $\rho$ exists.  We define the shorthand notations,
    \begin{align}
        W_{\Gamma,g} &= L_{\Gamma,g}\otimes R_{\Gamma,g} \\
        \mathcal{W}_{\Gamma,g} &= \mathcal{U} W_{\Gamma,g} \mathcal{U}^\dagger \\
        \tilde{\mathcal{W}}_{\Gamma,g} &= \mathcal{U}^\dagger W_{\Gamma,g} \mathcal{U}.
    \end{align}
    By multiplying both sides of Eq.~\eqref{eq:invariant} by $\rho$ and taking the trace, we have,
    \begin{equation} \label{eq:invariant_proof_1}
        \begin{aligned}
        c_{g,h} &=
             \tr(\mathcal{U}^\dagger U_{A,g}\mathcal{U}U_{B,h}\mathcal{U}^\dagger U_{A,g}^\dagger\mathcal{U}U_{B,h}^\dagger\rho) \\
             &=\tr(\mathcal{U}^\dagger U_{A,g}\mathcal{U}U_{B,h} \mathcal{U}^\dagger U_{A,g}^\dagger\mathcal{U}W_{B,h}^\dagger\rho ) \\
             &=\tr(\mathcal{U}^\dagger U_{A,g}\mathcal{U}U_{B,h} \mathcal{U}^\dagger U_{A,g}^\dagger \mathcal{W}_{B,h}^\dagger\rho \mathcal{U} ) \\
             &=c_{g,h} \tr(\mathcal{U}^\dagger U_{A,g}\mathcal{U}U_{B,h} \mathcal{U}^\dagger  \mathcal{W}_{B,h}^\dagger U_{A,g}^\dagger\rho\mathcal{U} ), \\
        \end{aligned}
    \end{equation}
    where we used the fact that $U^\dagger_{A,g}\mathcal{W}_{B,h}^\dagger U_{A,g}\mathcal{W}_{B,h} = c_{g,h} \id$ \cite{Zhang2023topological}. Continuing, we have,
    \begin{equation} \label{eq:invariant_proof_2}
        \begin{aligned}
            c_{g,h} 
             &=c_{g,h} \tr(\mathcal{U}^\dagger U_{A,g}\mathcal{U}U_{B,h} \mathcal{U}^\dagger  \mathcal{W}_{B,h}^\dagger W_{A,g}^\dagger\rho \mathcal{U} ) \\
             &=c_{g,h} \tr(\mathcal{U}^\dagger U_{A,g}\mathcal{U}U_{B,h}   \tilde{\mathcal{W}}_{A,g}^\dagger W_{B,h}^\dagger\rho ) \\
             &=c_{g,h}c_{h^{-1},g} \tr(\mathcal{U}^\dagger U_{A,g}\mathcal{U} \tilde{\mathcal{W}}_{A,g}^\dagger\rho ). \\
        \end{aligned}
    \end{equation}
    Therein, we used the fact that $\tilde{\mathcal{W}}_{A,g}^\dagger W_{B,h}^\dagger= W_{B,h}^\dagger \tilde{\mathcal{W}}_{A,g}^\dagger$ since they are supported in disjoint regions. We 
    then noticed that $U_{B,h} = U_{B,h^{-1}}^\dagger$ (since $U_{B,h}$ is a unitary representation of $G$) and again used the fact that $ U_{B,h^{-1}}^\dagger\tilde{\mathcal{W}}_{A,g}^\dagger U_{B,h^{-1}}\tilde{\mathcal{W}}_{A,g} = c_{h^{-1},g} \id$ \cite{Zhang2023topological}. Then we used symmetry localization to write $U_{B,h}W_{b,h}^\dagger\rho = \rho$.
    Finally, we have,
    \begin{equation} \label{eq:invariant_proof_3}
        \begin{aligned}
            c_{g,h} &=c_{g,h}c_{h^{-1},g} \tr(\mathcal{U}^\dagger U_{A,g} W_{A,g}^\dagger\rho \mathcal{U} ) \\
             &=c_{g,h}c_{h^{-1},g} \tr(\mathcal{U}^\dagger \rho \mathcal{U} )\\
             &=c_{g,h}c_{h^{-1},g}.
        \end{aligned}
    \end{equation}
    Therefore, $c_{h^{-1},g}=1$ for all $g,h\in G$.
\end{proof}

Similar invariants exist for non-abelian phases and in higher dimensions \cite{Zhang2023topological}, and we believe they can be used to prove the same result in those cases.

This theorem strongly constrains pure state catalysts. This is because any pure state that has exponentially decaying correlations must have strong symmetry localization
\begin{Cor} \label{cor:pure}
    Suppose a 1D state $|\psi\rangle$ is $(G,\mathcal{U})$-symmetric where $G$ is finite-abelian and $\mathcal{U}$ is an entangler for a non-trivial SPT phase. Then $|\psi\rangle$ must have some (quasi) long-range two-point correlation functions.
\end{Cor}
\begin{proof}
    We prove the contrapositive. If $|\psi\rangle$ has only exponentially decaying correlations, then it has area law entanglement \cite{Brandao2014,Cho2018} and can therefore be faithfully represented as a matrix product state (MPS) \cite{Verstraete2006}. But, every MPS with exponentially decaying correlations exhibits symmetry localization \cite{PerezGarcia2008}. From Theorem \ref{thm:loc}, we see that $|\psi\rangle$ cannot be $(G,\mathcal{U})$-symmetric where $\mathcal{U}$ is an entangler for a non-trivial SPT phase.
\end{proof}
This result for pure states is essentially equivalent to that of Ref.~\cite{Chen2011two}, which holds for general groups $G$, where the authors arrived at the constraint from the perspective of anomalous symmetries.

For mixed-state catalysts, the situation is different. As we have seen, the SW-SSB catalyst in Eq.~\eqref{eq:swssb_cat} has all short-range correlations (but long-range fidelity correlations). While the SW-SSB catalyst does not have strong symmetry localization, it does have weak symmetry localization. This means that, for any subregion $\Gamma$, there exist operators $L_{\Gamma,g}$ and $R_{\Gamma,g}$ such that
\begin{equation} \label{eq:weak_loc}
    U_{\Gamma,g}^\dagger\rho U_{\Gamma,g} = (L_{\Gamma,g}^\dagger\otimes R_{\Gamma,g}^\dagger)\rho(L_{\Gamma,g}\otimes R_{\Gamma,g}).
\end{equation}
This can be seen in the example of Eq.~\eqref{eq:swssb_cat} where $L_{\Gamma,g}=R_{\Gamma,g} = \id$. 

We can argue that any catalyst with weak symmetry localization must have a long-range fidelity correlator. Since $\mathcal{U}$ is symmertic and locality preserving, we can write, 
\begin{equation}
    \mathcal{U} U_{\Gamma,g} (L_{\Gamma,g}^\dagger \otimes R_{\Gamma,g}^\dagger) \mathcal{U}^\dagger = U_{\Gamma,g} (\mathcal{L}_{\Gamma,g}^\dagger \otimes \mathcal{R}_{\Gamma,g}^\dagger),
\end{equation}
where $\mathcal{L}_{\Gamma,g}$ and $\mathcal{R}_{\Gamma,g}$ are localized near the left and right endpoints of $\Gamma$. If $\mathcal{U}$ is the entangler of a non-trivial abelian SPT and $G$ is finite-abelian, then the operators $\mathcal{L}_{\Gamma,g}$ will carry different charge than $L_{\Gamma,g}$ for some $g\in G$ (likewise for the right endpoints). Then, we can write,
\begin{equation} \label{eq:weak_fid_proof}
    \begin{aligned}
        \rho &= (L_{\Gamma,g}\otimes R_{\Gamma,g})U_{\Gamma,g}^\dagger\rho U_{\Gamma,g}(L_{\Gamma,g}^\dagger\otimes R_{\Gamma,g}^\dagger) \\
        &= (\mathcal{L}_{\Gamma,g}\otimes \mathcal{R}_{\Gamma,g})U_{\Gamma,g}^\dagger\rho U_{\Gamma,g}(\mathcal{L}_{\Gamma,g}^\dagger\otimes \mathcal{R}_{\Gamma,g}^\dagger) \\
        & = (\mathcal{L}_{\Gamma,g}L_{\Gamma,g}^\dagger\otimes \mathcal{R}_{\Gamma,g}R_{\Gamma,g}^\dagger)\rho(L_{\Gamma,g}\mathcal{L}_{\Gamma,g}^\dagger\otimes R_{\Gamma,g}\mathcal{R}_{\Gamma,g}^\dagger),
    \end{aligned}
\end{equation}
where we conjugated by $\mathcal{U}$ and then applied Eq.~\eqref{eq:weak_loc}. The endpoint operators $\mathcal{L}_{\Gamma,g}L_{\Gamma,g}^\dagger$ and $\mathcal{R}_{\Gamma,g}R_{\Gamma,g}^\dagger$ carry non-trivial charge for some $g\in G$. Then Eq.~\eqref{eq:weak_fid_proof} implies a long-range fidelity correlator \eqref{eq:fidelity_corr} for charged operators, as claimed.

For general mixed-state catalysts, we make the following conjecture as an extension of Corollary \ref{cor:pure}:
\begin{conjecture}
\label{Conj1}
A catalyst $\rho$ of 1D SPT should have a (quasi) long-range correlation in the sense of R\'enyi-$n$ correlator for some $n\in\mathbb{Z}^+$ and local operator $O_i$, namely
\begin{align}
R^{(n)}(i,j)=\frac{\tr\left(\rho \sigma^{n-1}\right)}{\tr(\rho^n)},\quad \sigma=O_i^\dag O_j \rho O_j^\dag O_i.
\end{align}
\end{conjecture}
We note that the converse of this conjecture can be easily argued. First, note that a catalyst $\rho$ for a 1D SPT cannot be symmetrically-invertible, meaning there cannot exist a state $\tilde{\rho}\in \tilde{\mathcal{H}}$ and a strongly symmetric finite-depth local quantum channel $\E$, such that
\begin{align}
\E[\rho\otimes\tilde{\rho}]=\ket{0}\bra{0}\in\mathcal{H}\otimes\tilde{\mathcal{H}}.
\end{align}
This follows from the fact that symmetric invertibility implies strong symmetry localization \cite{Average_2023}, which we have shown is incompatible with being an SPT catalyst. Now, it was shown in Ref.~\cite{SWSSB2024} that all symmetrically invertible states have exponentially-decaying R\'enyi-$n$ correlators for all $n$. Our conjecture is therefore the converse, namely that any state which is not symmetrically invertible must have a long-range R\'enyi-$n$ correlator for some $n$.

\section{Applications to state preparation} \label{sec:state_prep}

Here we discuss how the many-body catalyst framework can be used to derive novel methods for preparing SPT-ordered states. Throughout this discussion, we assume that we are in a scenario where we only have access to symmetric gates or interactions since SPT phases can be easily prepared using non-symmetric FDQCs \cite{Chen2013}. Such a symmetry constraint may come from inherent properies of the underlying physical system, such as conservation of charge or particle number, or gauge symmetries in lattice gauge theories \cite{Schuch2004,Verstraete2006}.
It can also be desirable to engineer a system where only symmetric interactions occur since this can increase the robustness of certain models of quantum computation \cite{Miyake2010,Else2012,Raussendorf2019} or the efficacy of error correction (where we view the symmetry constraint as being similar to a biased-noise error model \cite{Tuckett2019,BonillaAtaides2021}).

At the most basic level, any efficient preparation unitary $\mathcal{U}_a$ for a catalyst $|a\rangle$ from a trivial product state immediately gives an efficient protocol to prepare any state $|\psi\rangle$ that it catalyzes. This is because we can simply prepare the catalyst in an ancillary system, then use it to prepare $|\psi\rangle$ via a symmetric FDQC,
\begin{equation}
    |\triv\rangle\otimes|\triv\rangle \xrightarrow{I\otimes \mathcal{U}_a} |\triv\rangle\otimes |a\rangle \xrightarrow{\mathcal{V}} |\psi\rangle\otimes |a\rangle,
\end{equation}
and we can additionally apply $I\otimes \mathcal{U}_a^{-1}$ at the end if we want the ancillary system to return to the trivial state. If $\mathcal{U}_a$ can be implemented in time $\tau$ using symmetric interactions, then $|\psi\rangle$ can be made in time $\tau + C$ where $C$ is the constant time needed to implement the FDQC $\mathcal{V}$ (or time $2\tau+C$ if we want to un-make the catalyst at the end). 

It turns out that  for the $(G,\mathcal{U})$-symmetric catalysts described in Theorem~\ref{thm:catalyst}, it is not necessary to use ancillary degrees of freedom to take advantage of the catalyst. Suppose a state $|\psi\rangle$ can be prepared using a symmetric QCA $\mathcal{U}$ and there is a catalyst $|a\rangle$ that is $(G,\mathcal{U})$-invariant. Then, consider the following process,
\begin{equation} \label{eq:4step}
    |\psi\rangle = \mathcal{U}\mathcal{U}_a^\dagger \mathcal{U}^\dagger \mathcal{U}_a |\triv\rangle.
\end{equation}
Since $\mathcal{U}$ is not generally a symmetric FDQC, this sequence of four evolutions does not describe an efficient preparation of $|\psi\rangle$. However, if we write $\mathcal{U}_a$ as a Hamiltonian evolution,
\begin{equation}
    \mathcal{U}_a = \mathcal{T} \left[e^{i\int_{0}^\tau H_a(t) dt}\right],
\end{equation}
for some time-dependent symmetric Hamiltonian $H_a(t)$, where $\mathcal{T}$ is the time-ordering operator, then we can rewrite Eq.~\eqref{eq:4step} as,
\begin{equation}
    |\psi\rangle = \mathcal{T} \left[e^{-i\int_{0}^\tau \tilde{H}_a(t) dt}\right]\mathcal{T} \left[e^{i\int_{0}^\tau H_a(t) dt}\right]|\triv\rangle.
\end{equation}
\diff{
where we have defined the modified Hamiltonian,
\begin{equation}
    \tilde{H}_a(t) = \mathcal{U} H_a(t) \mathcal{U}^\dagger.
\end{equation}
}
If $H_a(t)$ consists of symmetric terms then so will $\tilde{H}_a(t)$ since $\mathcal{U}_\spt$ is a symmetric QCA. Furthermore, $H_a(t)$ and $\tilde{H}_a(t)$ will have the same spatial locality. \diff{Note that this only holds if $H_a$ is a symmetric Hamiltonian, which is where the symmetry of the catalyst comes in.} So, if $|a\rangle$ can be prepared in time $\tau$ using a symmetric Hamiltonian evolution, then $|\psi\rangle$ can be prepared in time $2\tau$ using a symmetric Hamiltonian evolution with the same locality, \diff{without requiring ancillary degrees of freedom}. If we instead write $\mathcal{U}_a$ as a symmetric quantum circuit, then $\mathcal{U} \mathcal{U}_a^\dagger\mathcal{U}^\dagger$ will be a symmetric quantum circuit with the same depth and locality (up to a constant).

\subsection{Generating SPT phases with long-range interactions}

The above method can be applied to construct efficient methods to prepare SPT phases with long-range interactions. First, consider the case of all-to-all interactions. In Sec.~\ref{sec:1dlsm} we saw that the product of long-range Bell states defined in Eq.~\eqref{eq:lr_bell} can catalyze the 1D SPT fixed-point state. These Bell states can be prepared in constant time using a long-range two-body symmetric Hamiltonian, or equivalently using a symmetric FDQC with long-range two-qubit gates. Therefore, the 1D SPT can also be prepared in constant time using long-range symmetric interactions. A similar construction works for SPT phases in all dimensions. This reproduces the observation of Ref.~\cite{Stephen2024klocal} that all SPT phases can be prepared in constant time using long-range symmetric interactions.

We can also construct efficient methods to prepare SPT states using power-law decaying Hamiltonians. We define an $\alpha$-decaying Hamiltonian as one of the form,
\begin{equation}
    H = \sum_{i<j} h_{i,j},
\end{equation}
where $h_{i,j}$ acts on sites $i$ and $j$ and is bounded by $||h_{i,j}||<d_{i,j}^{-\alpha}$ where $d_{i,j}$ is the distance between $i$ and $j$ and $\alpha$ is some positive constant. The catalyst we consider is the GHZ state, which we have seen serves as a catalyst for SPT phases in all dimensions. Efficient protocols to prepare GHZ states using power-law interactions have been studied extensively due to their application in state-transfer, see Ref.~\cite{Chen2023speed} for a recent review. Using the resulting GHZ state as a catalyst, all of these protocols can be immediately translated protocols to prepare SPT phases~\footnote{It should be noted that the protocols used in optimal state-transfer protocols of Refs.~\cite{Eldredge2017,Tran2021} do not use symmetric Hamiltonians. However, it is straightforward to modify these protocols such that they are symmetric, without changing their efficiency. See, as a demonstration, the linear-depth circuit in Ref.~\cite{Chen2024sequential} for creating a GHZ state.}.

In general, the optimal time to make a GHZ state using an $\alpha$-decaying Hamiltonian is a complicated function of $\alpha$. Here, we mention as one example that, for $\alpha < D$ where $D$ is the spatial dimension, the time to make a GHZ state is independent of system size. Therefore, SPT phases can be made in constant time using $\alpha$-decaying symmetric Hamiltonians for $\alpha < D$. 
\diff{To the best of our knowledge, this is the first construction of symmetric long-range Hamiltonians for preparing SPT phases.}

\subsection{Generating SPT phases with local symmetric channels}

We showed in Sec.~\ref{sec:mixed} that SPT phases can be catalyzed using SW-SSB mixed states. Here we show how to prepare the mixed-states efficiently using only local symmetric measurements, thereby giving a constant-time local, symmetric channel to prepare SPT phases. Importantly, no long-range classical communication or unitary feedback is needed.

It is well-known that a GHZ state can be prepared in constant time using local measurements of the stabilizers and unitary operations conditioned on the measurement outcomes. However, those unitary corrections require long-range classical communication, and hence this protocol does not constitute a local, symmetric channel. Surprisingly, it turns out that the unitary correction is not necessarily needed to use the state as a catalyst because the post-measurement state can be $(G,\mathcal{U})$-invariant for all measurement outcomes.

Let us demonstrate this for the example of the 1D cluster state. Suppose we initialize an ancillary system in the symmetric product state $|+\rangle ^{\otimes N}$. Now, we measure the operator $Z_iZ_{i+2}$ for each $i$. Note that this measurement commutes with the $\Z_2\times \Z_2$ symmetry. The post-measurement state is,
\begin{equation}
    |a(\vec{s})\rangle = \prod_i \frac{\id+ s_i Z_iZ_{i+2}}{2}|+\rangle ^{\otimes N},
\end{equation}
where $\vec{s}=(s_1,\dots,s_N)$ and the signs $s_i=\pm 1$ are the measurement outcomes. Note that valid measurement outcomes satisfy $\prod_i s_{2i} = \prod_i s_{2i+1} = 1$ and all valid outcomes are equally likely. In the case that $s_i=1$ for all $i$, this describes a pair of GHZ states on the odd and even sublattices. For generic measurements, the post-measurement state is a GHZ state with some pattern of domain walls. These domain walls can be corrected by applying strings of local $X$ operators determined by the measurement outcomes. However, this correction step is unnecessary since the post-measurement state $|a(\vec{s})\rangle$ is a catalyst for all measurement outcomes. To see this, note that $|a(\vec{s})\rangle$ is a stabilizer state, and a minimal set of generating stabilizers with no relations between them is given by $U_e$, $U_o$, and $\{s_iZ_iZ_{i+2}, i=1,\dots,N-2\}$. Therefore, we can write,
\begin{equation}
    |a(\vec{s})\rangle\langle a(\vec{s})| = \prod_{i=1}^{N-2} \left(\frac{\id+ s_i Z_iZ_{i+2}}{2}\right)\left(\frac{\id + U_e}{2}\frac{\id + U_o}{2}\right),
\end{equation}
which is clearly invariant under conjugation by $\mathcal{U}_\CZ$. This implies that $\mathcal{U}_\CZ|a(\vec{s})\rangle = |a(\vec{s})\rangle$ up to an irrelevant phase.

Therefore, to prepare a cluster state, we can prepare a catalyst by performing symmetric local measurements on a product state and then apply the catalyzed symmetric FDQC described in Theorem~\ref{thm:catalyst}. The catalyst that we prepare on average in the process is 
\begin{equation}
\rho=\sum_{\vec{s}}p_{\vec{s}}|a(\vec{s})\rangle\langle a(\vec{s})| = \frac{1}{2^N}(\id + U_e)(\id + U_o),
\end{equation}
which is exactly the mixed-state catalyst described in Sec.~\ref{sec:mixed}. Therefore, the above process describes via an operational ``unravelling'' of what it means to use the mixed-state catalyst in practice. \diff{We stress that this mixed state is different from the pure GHZ state. In particular, it does not allow for long-range state transfer, so there is no contradiction that we are able to prepare it with a symmetric local channel.}

\section{Discussion} \label{sec:discussion}

In this work, we introduced the idea of a many-body quantum catalyst and demonstrated its richness and broad applicability in the context of SPT phases of matter. There are many avenues of future work, and we describe some here. First, it is worth noting that a majority of our results considered a particular class of catalysts that we called $(G,\mathcal{U})$-symmetric catalysts (Theorem \ref{thm:catalyst}), \diff{but we have not shown that all catalysts must belong to this class.} Can we find catalysts for SPT phases with new physical properties if we move away from this class? Or is it possible that all catalysts can be related in a natural way to a $(G,\mathcal{U})$-symmetric catalyst? In general, it is an important problem to constrain potential catalysts at a more basic level than we have done here. \diff{In this direction, Theorem \ref{thm:loc} gives constraints on potential catalysts that lie beyond the framework of Theorem \ref{thm:catalyst}.}

All of our results and examples considered bosonic systems. However, our ideas can be straightforwardly applied to fermionic systems. For example, the 1D Kitaev chain is an SPT phase protected by fermionic parity symmetry \cite{Kitaev2001}. It can be created from a trivial state using a symmetric QCA: the ``Majorana'' or ``half'' shift operator \cite{Po2017,Seiberg2024majorana}. Then, any fermionic state that is invariant under this operation can catalyze the Kitaev chain, see Refs.~\cite{Rahmani2015,OBrien2018,Seiberg2024majorana,Seiberg2024} for examples.
The discussion therefore exactly mirrors that given in Sec.~\ref{sec:1dlsm}, where the qubits therein are replaced by Majorana fermions. The fermionic case is especially interesting because the protecting parity symmetry cannot be explicitly broken. Therefore, unlike the case of bosonic SPT phases, there is no way to efficiently create a single Kitaev chain, making the catalyst approach potentially more appealing.

Another future direction involving fermions is to make a connection to the notion of {\it fragile topology} \cite{fragile}, wherein a topologically nontrivial energy band can be trivialized by the addition of trivial degrees of freedom. This represents an analog of our results for non-interacting fermions. Whether ideas introduced in our paper, possibly in combination with ideas from fragile topology, can be employed to develop an understanding of many body catalysts for interacting fermionic systems would be an interesting topic for future work. 

We have also focused on catalyzing short-range entangled phases of matter, and we suspect that catalysts will play a much different role when considering long-range entangled phases. As a preliminary comment, we remark that
some of the aspects of our story are reminiscent of the idea of foliated fracton phases of matter \cite{Shirley2018fracton}. Therein, certain 3D phases of matter called fracton phases have the property that an FDQC acting on the fracton state and a product state can turn that product state into a topologically ordered state. This process changes the 3D state by reducing its size in one direction, but preserves its phase of matter. So one can view a fracton phase as a catalyst for topologically ordered states, with the caveat that it can only be used finitely many times (assuming a finite system size).

In a similar vein, catalysts can also be used to transform between gapless phases of matter \cite{Scaffidi2017, Verresen2021gapless}. For example, it is shown in Ref.~\cite{Verresen2021gapless} that two gapless topological phases with distinct patterns of symmetry enrichment can be related by applying symmetric QCA corresponding to SPT entanglers. Since the catalysts we construct allow such QCA to be implemented on arbitrary states, they can also catalyze the transformation between certain symmetry-enriched gapless phases.

While we have only considered catalyzing transformations between pure states, possibly with use of mixed-state catalysts, we can also consider transformations between mixed states. One particularly relevant subject is SPT phases of mixed states and open quantum systems \cite{Roberts2017, Coser2019classificationof, de_Groot_2022, ma2024symmetry, lee2024symmetry, xue2024tensor, sun2024holographic}, which have also been called average SPT phases of matter \cite{aspt,zhang2024strange, Average_2023, guo2024LPDO, Zhang_2023fractonic}. It will be interesting to investigate whether transformations between such phases can be catalyzed, and whether they require FDQCs or more general quantum channels to implement.

Finally, we used the catalyst perspective to derive novel methods to prepare SPT phases using symmetric long-range interactions. These methods were based on known efficient methods to prepare GHZ states with power-law interactions. However, we know there are other kinds of catalysts including gapless states (and topological states in higher dimensions). Any efficient method to prepare these states would lead to efficient methods to prepare SPT phases. For gapless states in particular, we note that their preparation complexity appears similar to GHZ states with local interactions \cite{Ho2019}, but this may change with power-law interactions. Preparing gapless states with power-law interactions seems like an interesting problem with relations to Lieb-Robinson bounds \cite{Chen2023speed} and anomaly bootstrapping \cite{Lanzetta2023}.

\begin{acknowledgments}
DTS thanks Xie Chen, Aaron Friedman, Oliver Hart, Michael Hermele, Marvin Qi, Nathanan Tantivasadakarn, and Ruben Verresen for helpful discussions. JHZ thanks Yichen Xu for enlightening discussions. DTS is supported by the Simons Collaboration on Ultra-Quantum Matter, which is a grant from the Simons Foundation (651440). JHZ is supported by the U.S. Department of Energy under Award Number DE-SC0024324.
\end{acknowledgments}

\bibliography{biblio.bib}

\appendix

\section{Properties of cocycle states} \label{sec:cocycle_details}

Here we prove two claims for Sec.~\ref{sec:global_spt} about cocycle states, namely that the cocycle $\nu_d$ can be chosen such that (1) $\mathcal{U}(\nu_d)^L = \id$ for some $L$ and (2) $\mathcal{U}(\nu_d)|gg\dots g\rangle = |gg\dots g\rangle$ for all $g\in G$. First, let us be more precise about what we mean by saying such a choice is possible. As discussed shortly, cocycles can be sorted into equivalence classes such that cocycle states who cocycles belong to the same class lie in the same SPT phase. 
To create a catalyst for a desired cocycle state, it is sufficient to instead find a catalyst for any other cocycle state belonging to the same phase, since that state can be mapped to the desired state by a symmetric FDQC. Our claim is that every equivalence class contains a cocycle $\nu_d$ such that conditions (1) and (2) hold.

To prove this, we need to define the circuit $\mathcal{U}(\nu_d)$ more precisely.  Given a cocycle $\nu_d$, we define the gate \cite{Chen2013},
\begin{equation} \label{eq:cocycle_gate}
    U(\nu_d) = \sum_{g_1,\dots, g_d\in G} \nu_d(e,g_1,\dots,g_d)|g_1\cdots g_d\rangle\langle g_1\cdots g_d|.
\end{equation}
Then, given a triangulation $\Lambda_d$ consisting of $d$-simplices $\Delta_d$, we have,
\begin{equation}
    \mathcal{U}(\nu_d) = \prod_{\Delta_d\in \Lambda_d} U(\nu_d)^{s(\Delta_d)},
\end{equation}
where $U(\nu_d)$ acts on the $d+1$ sites included in the simplex $\Delta_d$. The signs $s(\Delta_d) = \pm 1$ depend on the orientation of the simplex $\Delta_d$ \cite{Chen2013}. 
In 2D, for instance, we can act with the gate $U(\nu_3)^{\pm 1}$ on the three qubits around every face of a triangular lattice where the sign depends on the orientation of the triangle. The gates $U(\nu_d)$ are all diagonal and hence commute, so $\mathcal{U}(\nu_d)$ is an FDQC.

We first argue that we can choose $\nu_d$ such that $\mathcal{U}(\nu_d)^L=\id$. We recall some basic notions in group cohomology \cite{Chen2013}. Any function $\lambda_d:G^{d+1}\rightarrow U(1)$ is called a $d$-cochain and the $d$-cochains form an abelian group under pointwise multiplication. Next, there exist certain coboundary operators $\partial_d$ that map from $d$-cochains to $d+1$-cochains. A cochain $\nu_d$ is called a cocycle if $\partial_d\nu_d = 1$ and it is called a coboundary if $\nu_d = \partial_{d-1}\mu_{d-1}$ for some $d-1$-cochain $\mu_{d-1}$ (note that every coboundary is automatically a cocycle). The cohomology classes $[\nu_d]$ are equivalence classes of cocycles modulo coboundaries and they form a group $H^{d}(G,U(1))$ called a cohomology group. 

Now, it is known that the cohomology groups are finite abelian groups. Therefore, for any cocycle $\nu_d$, there exists some $L$ such that $\nu_d^L$ belongs to the trivial class, meaning $\nu_d^L = \mu_d$ for a coboundary $\mu_d$. Now define $\tilde{\nu}_d = \nu_d \mu_d^{1/L}$ such that $\tilde{\nu}_d^L = 1$. This is again a cocycle because $\mu_d^{1/L}$ is a coboundary if $\mu_d$ is. We also have $[\nu_d]=[\tilde{\nu_d}]$. Then, from the definition in Eq.~\eqref{eq:cocycle_gate} it is clear that $U(\tilde{\nu}_d)^L = \id$ so $\mathcal{U}(\tilde{\nu}_d)^L=\id$ as well.

Next we prove that we can choose $\nu_d$ such that $\mathcal{U}(\nu_d)|gg\dots g\rangle = |gg\dots g\rangle$. From the definition in Eq.~\eqref{eq:cocycle_gate}, this is true if $\nu_d(e,g,\dots,g) = 1$. This can always be made true by modifying $\nu_d$ by an appropriate coboundary, as shown in Appendix J1 of Ref.~\cite{Chen2013}, which completes the proof.

\end{document}